\documentclass[oribibl]{llncs}

\usepackage{lineno,hyperref}
\usepackage{verbatim}
\usepackage{tikz}
\usepackage{graphicx}
\usepackage{caption}
\usepackage{subcaption}
\usepackage{clrscode3e}

\title{Global Value Numbering:\\ A Precise and Efficient Algorithm}
\author{Rekha R. Pai}
\institute{National Institute of Technology Calicut, Kerala, India \\ \email{rekharamapai@nitc.ac.in}}

\begin{document}

\maketitle

\begin{abstract}
Global Value Numbering (GVN) is an important static analysis to detect equivalent expressions in a program. We present an iterative data-flow analysis GVN algorithm in SSA for the purpose of detecting total redundancies. The central challenge is defining a \emph{join} operation to detect equivalences at a join point in polynomial time such that  later occurrences of redundant expressions could be detected. For this purpose, we introduce the novel concept of \emph{value $\phi$-function}. We claim the algorithm is precise and takes only polynomial time.
\end{abstract}

\section{Introduction} Global Value Numbering is an important static analysis to detect equivalent expressions in a program. Equivalences are detected by assigning \emph{value numbers} to expressions. Two expressions are assigned the same value number if they could be detected as equivalent. The seminal work on GVN by Kildall \cite{Kildall1973} detects all \emph{Herbrand equivalences} \cite{Ruething1999} in non-SSA form of programs using the powerful concept of \emph{structuring} but takes exponential time. Efforts were made to improve on efficiency in detecting equivalences. However the algorithms are either as precise as Kildall's \cite{Saleena2014} or efficient \cite{Ruething1999, Alpern1988, Gulwani2007} but not both.

The strive for combining precision with efficiency has motivated our work in this area. We propose an iterative data-flow analysis GVN algorithm to detect redundancies in SSA form of programs that is precise as Kildall's and efficient (i.e.\,take only polynomial time). As in a data-flow analysis problem, the central challenge is to define a \emph{join} operation to detect all equivalences at a join point in polynomial time such that any later occurrences of redundant expressions could be detected. We introduce the novel concept of \emph{value $\phi$-function} for this purpose.

\section{Terminology}
\paragraph{Program Representation} Input to our algorithm is the Control Flow Graph (CFG) representation of a program in SSA. The graph has empty \emph{entry} and \emph{exit} blocks. Other blocks contain assignment statements of the form $x = e$, where $e$ is an expression which is either a constant, a variable, or of the form $x \oplus y$ such that $x$ and $y$ are constants or variables and $\oplus$ is a generic binary operator. An expression can also be of the form $\phi_{k}(x, y)$, called \emph{$\phi$-functions}, where $x$ and $y$ are variables and $k$ is the block in which it appears. We assume a block can have at most two predecessors and a block with exactly two predecessors is called \emph{join} block. The input and output points of a block are called \emph{in} and \emph{out} points, respectively, of the block. The \emph{in} point of a join block is called \emph{join point}. We may omit the subscript $k$ in $\phi_{k}$ when the join block is clear from the context. In the CFGs we draw, $\phi$-functions appear in join blocks. But for clarity in explaining some of our concepts we assume $\phi$-functions are transformed to \emph{copy} statements and appended to appropriate predecessors of the join block.

\paragraph{Equivalence} Two expressions $e_{1}$ and $e_{2}$ are \emph{equivalent}, denoted $e_{1} \equiv e_{2}$, if they will have the same value whenever they are executed. Two expressions in a path are said to be \emph{equivalent in the path} if they are equivalent in that path. We detect only Herbrand equivalences \cite{Ruething1999} which is equivalence among expressions with same operators and corresponding operands being equivalent.

\section{Basic Concept}
Our main goal is to detect equivalences with a view to detecting redundancies in a program in polynomial time. We introduce the concept of \emph{value $\phi$-function} for the purpose which is explained in this section followed by our method to detect redundancies.

\subsection{Value $\phi$-function} Consider the simple code segment in Fig.\,1(a). Here irrespective of the path taken $x_{1}+y_{1}$ is equivalent to $a_{1}+b_{1}$. In terms of the variables being assigned to, we can say $z_{1}$ is equivalent to same variable $c_{1}$.
\begin{figure}[ht]
        \centering
        \begin{subfigure}[b]{0.3\textwidth}
                \includegraphics[width=17mm, height=15mm]{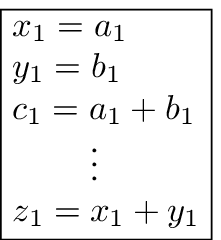}
                \caption{Linear program}
                \label{fig:1a}
        \end{subfigure}
        \begin{subfigure}[b]{0.4\textwidth}
                \includegraphics[width=55mm, height=20mm]{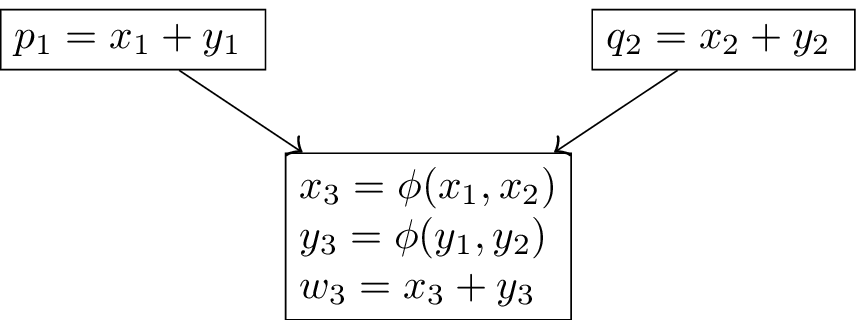}
                \caption{Program with branches}
                \label{fig:1b}
        \end{subfigure}
        \caption{Concept of value $\phi$-function}\label{fig:vpf}
\end{figure}
Now consider the code segment in Fig.\,1(b). Depending on the path taken expression $x_{3}+y_{3}$ is equivalent to either $x_{1}+y_{1}$ or $x_{2}+y_{2}$. In terms of the variables being assigned to, we can say $w_{3}$ is equivalent to merge of different variables -- $p_{1}$ and $q_{2}$. Inspired by the notion of $\phi$-function, we can say $w_{3}$ is equivalent to $\phi(p_{1}, q_{2})$. This notion of $\phi$-function is an extended notion of $\phi$-function as seen in the literature. In the literature, a $\phi$-function has different subscripted versions of the same non-SSA variable, say $\phi(x_{1}, x_{2})$. To express such equivalences, we introduce the concept of \emph{value $\phi$-function} similar to the concept of \emph{value expression} \cite{Saleena2014}.

\paragraph{Value $\phi$-function} A \emph{value $\phi$-function} is an abstraction of a set of equivalent $\phi$-functions (including the extended notion of $\phi$-function). Let $v_{i}$, $v_{j}$ be value numbers and \emph{vpf} be a value $\phi$-function. Then $\phi_{k}(v_{i}, v_{j})$, $\phi_{k}(\emph{vpf}, v_{j})$, $\phi_{k}(v_{i}, \emph{vpf})$, and $\phi_{k}(\emph{vpf}, \emph{vpf})$ are \emph{value $\phi$-functions}.

\paragraph{Partition} A partition at a point represents equivalences that hold in the paths to the point. An equivalence class in the partition has a value number and elements like variables, constant, and value expression. It is also annotated with a value $\phi$-function when necessary. The notation for a partition is similar to that in \cite{Saleena2014} except that a class can be annotated with value $\phi$-function.

\section{Proposed Method}
Using the concept of value $\phi$-function we propose an iterative data-flow analysis algorithm to compute equivalences at each point in the program. The two main tasks in this algorithm are \emph{join} operation and \emph{transfer function}:

\subsection{\emph{Join} operation.} A \emph{join} operation detects equivalences that are common in all paths to a join point. The join is conceptually a class-wise intersection of input partitions. Let $C_{1}$ and $C_{2}$ be two classes, one from each input partition. If the classes have same value number then the resulting class $C$  is intersection of $C_{1}$ and $C_{2}$. If the classes have different value numbers, say $v_{1}$ and $v_{2}$ respectively, then common equivalences are found by intersection of $C_{1}$ and $C_{2}$. The common equivalences obtained are actually a merge of different variables, which is indicated by the difference in value numbers and hence class $C$ is annotated with $\phi(v_{1}, v_{2})$. Now if the classes have different value expressions, say $v_{m}+v_{n}$ and $v_{p}+v_{q}$ respectively, the value expressions may be merged to form a resultant value expression say $v_{i}+v_{j}$. Value expressions $v_{m}+v_{n}$ and $v_{p}+v_{q}$ are merged to get $v_{i}+v_{j}$ by recursively merging classes of $v_{m}$ and $v_{p}$ to get class of $v_{i}$ and classes of $v_{n}$ and $v_{q}$ to get class of $v_{j}$ \cite{Saleena2014}. But merging the value expressions can lead to exponential growth of resulting partition \cite{Gulwani2007}. We do not merge different value expressions now instead merge them at a point where an expression represented by $v_{i}+v_{j}$ actually occurs in the program. This merge is achieved simply by detecting equivalence of $v_{i}+v_{j}$ with $\phi(v_{1}, v_{2})$ and is done during application of transfer function.

\paragraph{Example} Let us now consolidate the concept of join using an example. Consider the case of applying join on partitions $P_{1} = \{v_{1}, x_{1}, x_{3} | v_{2}, y_{1}, y_{3}, v_{1}+1 | v_{3}, z_{1}, z_{3}\}$ and $P_{2} = \{v_{4}, x_{2}, x_{3} | v_{5}, y_{2}, y_{3} | v_{6}, z_{2}, z_{3}, v_{4}+1 \}$. In the classes with value numbers $v_{1}$ in $P_{1}$ and $v_{4}$ in $P_{2}$ there is only one common variable $x_{3}$ and this will appear in a class in the resulting partition $P_{3}$. Since the two classes in $P_{1}$ and $P_{2}$ have different value numbers $v_{1}$ and $v_{4}$, respectively, the resulting class is annotated with value $\phi$-function $\phi(v_{1}, v_{4})$. The class is assigned a new value number, say $v_{7}$. The resulting class is $| v_{7}, x_{3} : \phi(v_{1}, v_{4}) |$. Now consider the classes with value numbers $v_{2}$ in $P_{1}$ and $v_{6}$ in $P_{2}$.  There are no obvious common equivalences in the classes and we don't merge the different value expressions now. Hence no new class is created. Similar strategies are adopted in detecting common equivalences in other pairs of classes one each from $P_{1}$ and $P_{2}$. The resulting partition $P_{3}$ is $\{v_{7}, x_{3} : \phi(v_{1}, v_{4}) | v_{8}, y_{3} : \phi(v_{2}, v_{5}) | v_{9}, z_{3} : \phi(v_{3}, v_{6})\}$.
 
 \begin{codebox}
 \Procname{$\proc{Join}(P_{1}, P_{2})$}
  \li $P \gets \{ \}$  
  \li \For each pair of classes $C_{i} \in P_{1}$ and $C_{j} \in P_{2}$
  \li    \Do
             $C_{k} \gets C_{i} \cap C_{j}$ \>\>\>\>\>\>\> \small \Comment set intersection
  \li        \If $C_{k} \neq \{\}$ and $C_{k}$ does not have value number
  \li        \Then
                 $C_{k} \gets C_{k} \cup \{v_{k}, \phi_{b}(v_{i}, v_{j})\}$ \>\>\>\>\> \small \Comment $v_{k}$ is new value number
  \zi   \>\>\>\>\small \Comment $v_{i} \in C_{i}$, $v_{j} \in C_{j}, $ $b$ is join block
             \EndIf
  \li    $P \gets P \cup C_{k}$ \>\>\>\>\>\>\> \small \Comment Ignore when $C_{k}$ is empty  
  \li \Return $P$
\end{codebox}
Note: We define special partition $\top$ such that $\proc{Join}(\top, P) \gets P \gets \proc{Join}(\top, P)$. We assume $\phi$-functions in a join block are transformed to copies and appended to appropriate predecessors of join block.

\subsection{Transfer Function.}\label{trfn} Given a partition $PIN_{s}$, that represents equivalences at \emph{in} point of a statement $s: x \gets e$ the transfer function computes equivalences at its \emph{out} point, denoted $POUT_{s}$. Let \emph{ve} be the value expression of $e$ computed using $PIN_{s}$. If \emph{ve} is present in a class in $PIN_{s}$, then $x$ is just inserted into corresponding class in $POUT_{s}$. Otherwise the transfer function checks whether $e$ could be expressed as a merge of variables represented by a value $\phi$-function \emph{vpf} (as illustrated below). If it is present in a class in $PIN_{s}$ then $x$, \emph{ve} are inserted into corresponding class in $POUT_{s}$. Else a new class is created in $POUT_{s}$ with new value number and $x$, \emph{ve}, \emph{vpf} are inserted into it. 

For an example, consider processing the statement $w_{3} = x_{3}+ y_{3}$ as shown in code segment in Fig.\,2.
\begin{figure}[ht]
 \centering
 \includegraphics[width=115mm, height=30mm]{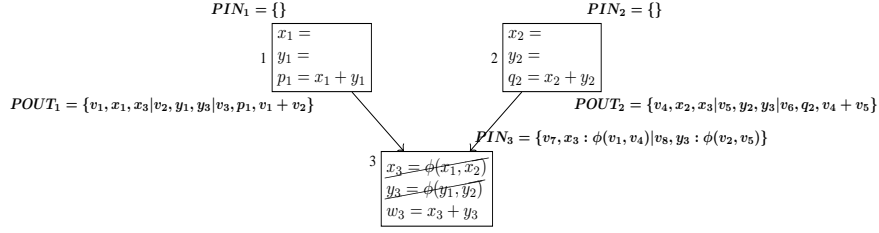}
\caption{Concept of Transfer Function}\label{fig:tf}
\end{figure} 
Since value expression $v_{7}+v_{8}$ of $x_{3}+y_{3}$ is not in $PIN_{3}$, the transfer function proceeds to check whether $x_{3}+y_{3}$ is actually a merge of variables as follows:\\
$x_{3}+y_{3} \equiv v_{7}+v_{8} \equiv \phi(v_{1}, v_{4})+\phi(v_{2}, v_{5}) \equiv \phi(v_{1}+v_{2}, v_{4}+v_{5}) \equiv \phi(v_{3}, v_{6}) $.\\
This implies $x_{3}+y_{3}$ is actually a merge of variables, here $p_{1}$ and $q_{2}$. Since neither $v_{7}+v_{8}$ nor $\phi(v_{3}, v_{6})$ are present in $PIN_{3}$, a new class is created in $POUT_{3}$ with new value number say $v_{9}$ and $w_{3}$, $v_{7}+v_{8}$, and $\phi(v_{3}, v_{6})$ are inserted into it. The classes in $PIN_{3}$ are inserted as such into $POUT_{3}$. The resulting partition $POUT_{3}$ is $\{v_{7}, x_{3} : \phi(v_{1}, v_{4}) | v_{8}, y_{3} : \phi(v_{2}, v_{5}) | v_{9}, w_{3}, v_{7}+v_{8} : \phi(v_{3}, v_{6}) \}$.
\begin{codebox}
 \Procname{$\proc{transferFunction}(x \gets e, PIN_{s})$}
 \li $POUT_{s} \gets PIN_{s}$
 \li $C_{i} \gets C_{i} - \{x\}$	\>\>\>\>\>\>\>\> \small \Comment $x \in C_{i}$, a class in $POUT_{s}$
 \li \emph{ve} $\gets \proc{valueExpr}(e)$
 \li \emph{vpf} $\gets \proc{valuePhiFunc}(ve, PIN_{s})$ \>\>\>\>\>\>\>\> \small \Comment can be NULL
 \li \If \emph{ve} or \emph{vpf} is in a class $C_{i}$ in $POUT_{s}$ \>\>\>\>\>\>\>\> \small \Comment ignore \emph{vpf} when NULL
 \li \Then
       $C_{i} \gets C_{i} \cup \{x, ve\}$  \>\>\>\>\>\> \small \Comment set union
 \li \Else
       $POUT_{s} \gets POUT_{s} \cup \{v_{n}, x, ve :$ \emph{vpf}$\}$ \>\>\>\>\>\> \small \Comment $v_{n}$ is new value number  
     \EndIf
 \li \Return $POUT_{s}$
\end{codebox}
The \proc{valuePhiFunc} is a recursive algorithm to compute value $\phi$-function corresponding to input value expression when possible else it returns NULL.

\subsection{Detect Redundancies.} Given partition $POUT$ at \emph{out} of statement $x = e$, expression $e$ is detected to be redundant if there exists a variable in the class of $x$ in $POUT$, other than $x$, or the class of $x$ in $POUT$ is annotated with value $\phi$-function. In the example code in Fig.\,2, consider the case of checking whether $x_{3}+y_{3}$ in the last statement $w_{3} = x_{3}+y_{3}$ is redundant. In the class of $w_{3}$ in $POUT_{3}$ (computed in previous subsection) there are no variables other than $w_{3}$. However the class is annotated with a value $\phi$-function. Hence the expression $x_{3}+y_{3}$ is detected to be redundant.

\begin{theorem}
 Two program expressions are equivalent at a point iff the iterative data-flow analysis algorithm detects their equivalence.
\end{theorem}
\begin{proof}
 This can be proved by induction on the length of a path in a program. \qed
\end{proof}

\section{Complexity Analysis}
Let there be $n$ expressions in a program. The two main operations in this iterative algorithm are join and transfer function. By definitions of \proc{Join} and \proc{transferFunction} a partition can have $O(n)$ classes. If there are $j$ join points, the total time taken by all the join operations in an iteration is $O(n.j)$. The transfer function involves constructing and then looking up for value expression or value $\phi$-function in the input partition. The transfer function of a statement takes $O(n.j)$ time. In an iteration total time taken by transfer functions is $O(n^{2}.j)$. Thus the time taken by all the joins and transfer functions in an iteration is $O(n^{2}.j)$. In the worst case the iterative analysis takes $n$ iterations and hence the total time taken by the analysis is $O(n^{3}.j)$.

\section{Conclusion}
We presented GVN algorithm using the novel concept of value $\phi$-function which made the algorithm precise and efficient.

\bibliographystyle{splncs03}
\bibliography{/home/rekha/my_Home/research/reports/my_bib.bib}
\end{document}